\documentclass[10pt]{article}
\usepackage{amssymb,latexsym,amsmath,amscd,amsthm}
\usepackage{amsmath,amssymb} 
\usepackage{amsthm} 
\usepackage{hyperref}
\usepackage{graphicx}
\usepackage{epsfig}
\usepackage{psfrag}
\usepackage{enumerate}
\usepackage{color}
\theoremstyle{definition}
\newtheorem{defi}{Definition}[section]
\newtheorem{ex}[defi]{Example}
\theoremstyle{plain}
\newtheorem{lema}[defi]{Lemma}
\newtheorem{thm}[defi]{Theorem}

\newtheorem{cor}[defi]{Corollary}
\newtheorem{rk}[defi]{Remark}

\definecolor{pink}{rgb}{1,0,1}

 \definecolor{blue}{rgb}{0,0,1}

 \definecolor{garnet}{RGB}{210,15,30}

\newtheorem{teo-def}[defi]{Theorem/Definition}

\newcommand{\CC}{\mathbb{C}}
\newcommand{\KK}{\mathbb{K}}
\newcommand{\RR}{\mathbb{R}}

\newcommand{\dd}{\mathrm{diag}}
\newcommand{\EE}{\mathcal{E}}

\usepackage{anysize} 

\marginsize{2.5cm}{2.5cm}{1.5cm}{3cm}

\newif\ifprivate
\privatetrue

\def\???{\ifprivate {\bf {???}} \marginpar{{\Huge {\bf ?}}}
\else \fi}

 \definecolor{zielony}{rgb}{0.5, 0.9, 0.1}
 \definecolor{czerwony}{rgb}{0.9, 0.2, 0.1}
 \definecolor{niebieski}{rgb}{0.3, 0.1, 0.9}

\title{Embeddability of Kimura 3ST Markov matrices}
\author{Jordi Roca-Lacostena and Jes\'us Fern\'andez-S\'anchez}

\begin{document}

\maketitle

\begin{abstract}
In this note, we characterize the embeddability of generic Kimura 3ST Markov matrices  in terms of their eigenvalues. As a consequence, we are able to compute the  volume of such matrices relative to the volume of all Markov matrices within the model.  %
We also provide examples showing that, in general, mutation rates are not identifiable from substitution probabilities.
These examples also illustrate that symmetries between mutation probabilities do not necessarily arise from symmetries
between the corresponding mutation rates.
\end{abstract}


\section{Introduction}

Genomic data expressed by means of sequence alignments is widely used to infer phylogenetic relationships between species. Substitution models are used to describe the evolutionary process that leads from one DNA sequence to another. These models are usually given in terms of a family of Markov matrices with a prescribed structure. The entries of these matrices represent the conditional probabilities of nucleotide substitution between one sequence and the other, and can be obtained either by counting the relative frequencies of these substitutions or fitting the parameters of the model using maximum likelihood. Usually, the structure imposed by the model is motivated by some biological / biochemical properties observed (e.g. the Kimura 3ST model \cite{kimura1981}) or some computational / mathematical convenient assumptions to deal with the model (e.g the GTR model \cite{tavare} or Lie Markov models \cite{LMM}).
Moreover, evolution is usually modelled by means of Markov chains, together with the additional assumption that all sites in the sequences evolve independently and according to the same probabilities. 

A general approach in modelling evolution corresponds to regarding time as a continuous variable where substitution events always happen at the same rate, which remains constant throughout the whole evolutionary process. This leads to the homogeneous continuous-time substitution models, where only  Markov matrices that are the exponential of a rate matrix are considered.
Clearly this is used as an approximation to biological reality where it is well known that transition rates vary over time \cite{ho2005,ho2007} and also among the different branches of the phylogenetic tree \cite{lockhart1998}.
However, given the bias / variance compensation of the statistical analysis \cite{burnham2002}, modelling phylogenetic evolution as a non-homogeneous process is not statistically feasible in practice (cf. \cite{LMM}).

A different approach appears when one regards the evolutionary process as a whole and only takes into account the conditional probabilities between the original and the final sequences, without caring about rates of mutation. When these probabilities are taken as the parameters of the model, we deal with the so-called \emph{algebraic} models\footnote{Here, ``algebraic'' refers to the fact that the probabilities of pattern observation at the leaves of a phylogenetic tree evolving under these models are given by algebraic expressions (only sums and products) in terms of the parameters of the model.}. 
Algebraic models have been used in a number of theoretical papers, including  \cite{allman2004,strumfels2004,draisma2008,casfer2010}. 
If one attempts to connect both approaches, a natural question is to decide whether a given Markov matrix is the exponential of some rate matrix, whose entries would be some kind of average of the rates involved throughout the evolutionary process. 
%
In this case, we say the matrix is \emph{embeddable} and this question is known in the literature as the \emph{embedding problem} for Markov matrices. 
An easier version of this problem is to decide whether the rate matrices associated to the embeddable matrices of a particular (algebraic) model $\mathcal M$ should keep the same symmetries as the model ($\mathcal M$-\emph{embeddability}, see definition in Section 2.2).
The embedding problem is relevant even if restricted to continuous-time models since it is not true in general that the product of embeddable matrices is necessarily embeddable (indeed, the Baker-Campbell-Hausdorff formula \cite{campbell1897} leads to ask whether some series of matrices is convergent or not, which is not always true \cite{casasblanes}).
These questions are closely related to the problem of the multiplicative closure of continuous-time models, namely whether the product of matrices $e^{Q_1} e^{Q_2}$ where $Q_1$ and $Q_2$ are rate matrices in one particular (continuous-time) model can be obtained as some $e^{Q}$ for some rate matrix $Q$ in \emph{the same model}. 
After \cite{LMM,Sumner2017}, it is known that there are popular models which are not multiplicatively closed, notably including the GTR model and the HKY model. 

The reader is referred to \cite{Davies} for a nice overview of the embedding problem from a mathematical point of view. 
%
In a more biological and applied setting, the paper by Verbyla \emph{et al.} \cite{verbyla} deals with the possible consequences for phylogenetic inference. 
Also, the paper \cite{Sumner2012} and the more recent paper \cite{Sumner2017b} deal with the incidental question of how the lack of (multiplicative) closure in substitution models have consequences for the phylogenetic analysis of data. 
%

In this paper, we deal with the embedding problem from a theoretical perspective. The main goal is to obtain a characterization for the embeddability of generic matrices of the Kimura 3ST model \cite{kimura1981}. From our results, we will be able to compute the whole volume of embeddable Kimura 3ST matrices and compare it with the volume of the whole space of Kimura 3ST Markov matrices. 
%
At the same time, we provide a number of examples showing matrices that are embeddable but for which 
the mutation rates are not identifiable or do not keep the same structure of the model. 
The recent paper \cite{Kaie} deals with the similar question of characterizing embeddable matrices of symmetric group-based phylogenetic models, but focusing on the existence of rate matrices strictly in the model. 

The organization of the paper is as follows. In section 2, we recall some definitions and basic facts concerning the embedding problem and the Kimura 3 parameter model.
Here, we also show that any embeddable matrix is biologically relevant since it can be seen as the transition matrix of a concatenation of \emph{realistic} evolutionary processes (``realistic'' here means a process whose transition matrix is close to the identity matrix, see Theorem \ref{realism}). In section 3, we prove the main theorem which characterizes under the (generic) assumption of having different eigenvalues the Kimura 3ST embeddable matrices in terms of inequalities to be satisfied by the eigenvalues. 
We devote as well some attention to the case of matrices with repeated eigenvalues as they present certain situations that may be interesting from a theoretical and applied point of view.  
Namely, these matrices show that the identifiability of the mutation rates is not a generic property for the Kimura 2ST model or the Jukes-Cantor model, as well as that there are embeddable matrices with rate matrices that do not keep the same symmetries of the model (see Theorem \ref{thm:identif}).
As a consequence of the characterization mentioned above, in section 4 we are able to compute the volume of embeddable matrices and compare it to the volume of all Kimura 3ST Markov matrices. Finally, Section 5 discusses implications and possibilities for future work.


\section{Preliminaries}

\subsection{Embedding problem of Markov matrices}
We denote by $M_k(\KK)$ the space of all square $k$-matrices with entries in a field $\KK$, where $\KK$ is $\RR$ or $\CC$. 
Given a matrix $A\in M_k(\mathbb{K})$,  we say that $B\in M_k(\mathbb{K})$ is a \emph{logarithm} of $A$ if $e^B=A$, where the exponential of a matrix is defined as 
\begin{eqnarray*}
 e^X=\sum_{n\geq 0} \frac{X^n}{n!}.
\end{eqnarray*}
A classical result states that $det(e^X)=e^{tr(X)}$, so the determinant of any matrix of the form $e^X$ is never 0. 
Given a non-negative complex number $x\in \CC\setminus \RR^{-}$, we will denote by $log(x)$ its \emph{principal logarithm}, that is, the only logarithm of $x$ that lies in the strip $\{z \mid -\pi < Im(z) < \pi \}$. 
Although the exponential map of matrices is not injective, it is known that if $A$ is a matrix with no negative eigenvalues, there is a unique logarithm $X$ of $A$ all of whose eigenvalues are given by the principal logarithm of the eigenvalues of $A$ (Theorem 1.31 of \cite{higham}). We will refer to this as the \emph{principal logarithm of} $A$ and we will denote it by $Log(A)$. 
%
%
In the particular case where the matrix $A$ is diagonalizable, $A=S\, D\,S^{-1}$ then $Log(A)=S Log(D) S^{-1}$, where $Log(D)$ is the diagonal matrix with diagonal entries equal to the principal logarithm of the eigenvalues of $A$.

\begin{defi}
 A matrix $M\in M_k(\mathbb{R})$ is said to be a \emph{Markov matrix} if all the entries are non-negative and the rows sum to one.
 A matrix $Q\in M_k(\mathbb{R})$ is said to be a \emph{rate matrix} if all the non-diagonal entries are non-negative and the rows sum to zero.
\end{defi}

If $Q$ is a rate matrix, it is well-known that $e^{tQ}=\sum_{n\geq 0} \frac{t^n Q^n}{n!}$ is a Markov matrix for all $t\geq 0$. That is why rate matrices are also referred as \emph{Markov generators} \cite{Davies}. However, not every Markov matrix can be obtained in this way. A Markov matrix $M$ is said to be \emph{embeddable} if $M=e^{Q}$ for some rate matrix $Q$. The \emph{embedding problem} attempts to decide which (Markov) matrices are embeddable, that is, which matrices can be written as $M=e^Q$, where $Q$ is a rate matrix. 
We would like to point out that every embeddable matrix can be obtained as the substituion matrix of a long-running biologically realistic Markov process. Namely,

%

\begin{thm}\label{realism}
 Every embeddable matrix is the product of embeddable matrices close to the identity matrix. 
\end{thm}

\begin{proof}
Assume that $M$ is an embeddable Markov matrix: $M=e^Q$. 
 Clearly, $Q_n:=\frac{1}{n}Q$ is still a rate matrix for any $n\geq 1$, so $M= (e^{Q_n})^n$ appears as the $n$-th power of a Markov matrix. Moreover, since
 \begin{eqnarray*}
  \lim_{n\rightarrow \infty} e^{Q_n} = e^{\lim_{n\rightarrow \infty} Q_n}= e^{(0)}=Id,
 \end{eqnarray*}
we can take $n$ big enough so that $e^{Q_n}$ is as close to $Id$ as wanted. 
 \end{proof}

\vspace{2mm}
\subsection{Kimura models}
In this work we deal with the substitution model introduced by Kimura in \cite{kimura1981}. The Kimura 3ST model assigns three parameters to different type of substitutions: one parameter for transitions, i.e. substitutions between purines ($\tt A \leftrightarrow \tt G$) or pyrimidines ($\tt C \leftrightarrow \tt T$), and two parameters for transversions, i.e. substitutions that change the type of nucleotide: from purine to pyrimidine or vice versa. 
Ordering the set of nucleotides as ${\tt A},{\tt G},{\tt C},{\tt T}$, the Markov matrices within the model are described by the following structure: 
\begin{defi}
A matrix $M\in M_4(\CC)$ is \textit{Kimura 3ST} (is K3 or has K3 form, for short) if it has the following structure: 
\begin{eqnarray}\label{K3form}
 M=\begin{pmatrix}
  a & b & c & d \\
  b & a & d & c \\
  c & d & a & b \\
  d & c & b & a 
  \end{pmatrix}.
  \end{eqnarray}
For ease of reading we will use the notation $M=K(a,b,c,d)$ to denote a matrix with the structure in~(\ref{K3form}). 
\end{defi}

When $M$ is a Markov matrix, the structure above describes the symmetry between the substitution probabilities of the Kimura 3ST model. Keeping the order of the nucleotides, if $i,j={\tt A, G, C, T}$, the $(i,j)$-entry corresponds to the probability of nucleotide $i$ being replaced by nucleotide $j$. 
The submodels of Kimura 3ST model, namely Kimura 2ST \cite{Kimura1980} and Jukes-Cantor \cite{JC69}, appear when more symmetries are considered: 
if $c=d$, we say that the matrix $M$ is \textit{Kimura 2ST} (K2, for short); if $b=c=d$, we say that $M$ is \textit{Jukes-Cantor} (JC, for short).

If one restricts the embedding problem to one particular model $\mathcal{M}$ given by some equalities between the entries of the Markov matrices (such as in Kimura 3ST, Kimura 2ST, Jukes-Cantor), it is natural to ask whether these matrices have Markov generators fulfilling the same equalities.
In this case, we say that these matrices are $\cal{M}$-\emph{embeddable}. %
Example \ref{ex35} of the next section shows that this is not true in general.

\vspace{3mm}
The following lemma is fundamental for our study. It essentially claims that all K3 matrices are diagonalized through the following Hadamard matrix: \begin{eqnarray*}
S:=\left ( \begin{array}{rrrr}
1 & 1 & 1& 1 \\ 1 & 1 & -1& -1\\1 & -1 & 1 & -1 \\1 & -1 & -1 & 1
 \end{array} \right ).
\end{eqnarray*}

Note that $S^2=4 \cdot Id$, thus $S^{-1}=\frac{1}{4} S$.

\begin{lema}\label{Lema:K3Diag}\
A matrix is K3 if and only if it can be diagonalized through $S$. %
In this case, 
$K(a,b,c,d)=S\, D\, S^{-1}$, where 
$D=\operatorname{diag}(a+b+c+d , \ a+b-c-d, a-b+c-d, \ a-b-c+d)$. 
 In particular, a K3 matrix is real if and only if its eigenvalues are all real numbers.
\end{lema}

\begin{proof}
The proof is straightforward and follows by direct computation. 
%
\end{proof}

Since the rows of K3 Markov matrices sum to $a+b+c+d=1$, the first eigenvalue must be equal to one. %
Hence, we derive that every K3 Markov matrix is determined by the set of the other eigenvalues: 
\begin{eqnarray}\label{coord}
x:=a+b-c-d, \qquad
y:=a-b+c-d, \qquad
z:=a-b-c+d. 
\end{eqnarray}
Moreover, it is immediate to check that a matrix $M$ as in (\ref{K3form}) is K2 (resp. JC) if and only if $y=z$ (resp. $x=y=z$).

It also follows that 
\begin{thm} \label{closureK3embed} 
The product of $K3$-embeddable matrices is $K3$-embeddable. 
\end{thm}

\begin{proof}
By virtue of Lemma \ref{Lema:K3Diag}, we have 
\[K(a,b,c,d)\cdot K(a',b',c',d')=( S\, D \, S^{-1}) \cdot ( S \, D'\, S^{-1})=S\, D D'\, S^{-1},\]
which is a K3 matrix.
Since $D$ and $D'$ are diagonal matrices, we have that $D\, D' = D'\, D$ and from this, the product of K3 matrices is commutative. 
Therefore, if $Q_1$, $Q_2$ are K3 rate matrices,  the Baker-Campbell-Hausdorff formula \cite{campbell1897} gives 
\[e^{Q_1} e^{Q_2}=exp \Big (Q_1+Q_2+\frac{1}{2}[Q_1,Q_2]+\dots \Big )=exp \big (Q_1+Q_2 \big ),\]
where $[A,B]=AB-BA$ is the Lie bracket.
Since $Q_1+Q_2$ is also a K3 rate matrix, the claim follows. 
\end{proof}

\section{Embeddability of Kimura Markov matrices}
 
 The first result of this section describes how to compute the principal logarithm of a K3 Markov matrix and characterizes when it is a Markov generator. 
 
First, we need a lemma which follows from the definition of the exponential matrix. 
 \begin{lema}\label{eigs_MQ}
Let $Q$ be a logarithm of $M$. If $v$ is an eigenvector of $Q$ with eigenvalue $\mu$, then $v$ is an eigenvector of $M$ with eigenvalue $e^{\mu}$. 
 \end{lema}

\begin{rk}\rm
Note that the converse is not true in general. For instance, any vector of $\mathbb{C}^2$ is an eigenvector of $M:= \left ( \begin{array}{cc} -1 & 0 \\ 0 & -1 \end{array} \right )$ with eigenvalue $-1$. However, the matrix  $Q=\left ( \begin{array}{cc} 0 & -1 \\ 1 & 0 \end{array} \right )$ is a logarithm of $M$, with eigenvalues $i$ and $-i$ and  corresponding eigenspaces $[(i,1)]$ and $[(-i,1)]$. In particular, any vector not in these subspaces cannot be an eigenvector of $Q$. 
\end{rk}

The following result solves the $K3$-embeddability.

\begin{thm}\label{thm:LogK3Embed}
Let $M$ be a K3 Markov matrix with eigenvalues $1,x,y,z$.
Then,
\begin{enumerate}[i)]
 \item $M$ has a real logarithm $Q$ with K3 form if and only if $x,y,z>0$. %
 In this case, $Q$ is necessarily the principal logarithm $Log(M) = S \; \dd(0,log(x),log(y),log(z)) \; S^{-1}$.  
 
 \item $M$ is $K3$-embeddable (i.e. $Log(M)$ is a Markov generator) if and only if   
\begin{eqnarray}
\label{ineq}
x\geq y z, \qquad   
y\geq x z, \qquad 
z\geq x y.
\end{eqnarray}
\end{enumerate}
 \end{thm}
 
\begin{proof}
(i) First of all, if $Q$ is a real logarithm of $M$ with K3 form, then the eigenvalues of $Q$ are real (Lemma \ref{Lema:K3Diag}) and, by Lemma \ref{eigs_MQ}, the eigenvalues of $M$ have to be positive. 
Conversely, if $x,y,z>0$ then we can take principal logarithms of $x,y,z$ and define $Q$ as the principal logarithm of $M$, i.e. $Q:=Log(M)=S \; \dd(0,log(x),log(y),log(z)) \; S^{-1}$. Then, $Q$ is a real matrix, $e^Q=M$ and it has K3 form because of Lemma \ref{Lema:K3Diag}. 

We proceed to show that $Log(M)$ is the only possible real logarithm with K3 form. By Lemma \ref{Lema:K3Diag}, if $Q'$ is any other real logarithm with K3 form, then it can be written as $Q'=S \; \dd(r,s,u,v) \; S^{-1}$, for some $r,s,u,v\in \mathbb{R}$. Because of Lemma  \ref{eigs_MQ}, these values are real logarithms of $1,x,y,z$. Necessarily, $Q'$ must be the principal logarithm of $M$. 

\vspace{2mm}
\noindent
(ii) 
%
%
Using Lemma \ref{Lema:K3Diag} we have $Log(M) = K(\alpha,\beta,\gamma,\delta)$, where $\alpha =  \frac{1}{4} \big (log(x)+log(y)+log(z) \big )$, 
$\beta =\frac{1}{4} \big (log(x)-log(y)-log(z)\big)$, 
$\gamma = \frac{1}{4} \big (log(y)-log(x)-log(z)\big)$ and 
$\delta = \frac{1}{4} \big (log(z)-log(x)-log(y)\big)$. 
It is immediate that $\alpha+\beta+\gamma+\delta=0$. Therefore, we only need to check that the non-diagonal entries $\beta$, $\gamma$ and $\delta$ of $Log(M)$ are non-negative if and only if the inequalities (\ref{ineq}) are satisfied. This is straightforward: for instance, 
\begin{eqnarray*}
\beta \geq 0  \Leftrightarrow  \frac{log(x)-log(y)-log(z) }{4}\geq0  \Leftrightarrow  log\left(\frac{x}{yz}\right) \geq 0  \Leftrightarrow  x\geq yz.
\end{eqnarray*}
The other inequalities are proved similarly. 
%
%
%
\end{proof}

\begin{lema} 
\label{thm:K3RealLog}
If $M$ is a K3 matrix with no repeated eigenvalues and $Q$ is a real logarithm of $M$, then $Q$ is the principal logarithm of $M$. In particular, $Q$ has $K3$ form. 
\end{lema}

 \begin{proof} 
Because of Lemma \ref{eigs_MQ}, if the matrix $M$ has no repeated eigenvalues, then so does the matrix $Q$. It follows that $Q$ diagonalizes, and that $M$ and $Q$ have the same eigenvectors. In particular they both diagonalize through the matrix $S$ and hence $Q$ must be K3 (Lemma \ref{Lema:K3Diag}). Now, it is enogh to apply Theorem \ref{thm:LogK3Embed}.
 %
  %
 \end{proof}

\begin{cor}\label{cor:K3noRepeated}
Let $M$ be a K3 Markov matrix with no repeated eigenvalues. Then, the following are equivalent: 
\begin{enumerate}[(i)]
 \item $M$ is  embeddable;
 \item $M$ is  $K3$-embeddable;
 \item the eigenvalues $x,y,z$ of $M$ are strictly positive, and satisfy 
 \begin{eqnarray*}
  x\geq y z, \qquad 
  y\geq x z, \qquad 
  z\geq x y.
 \end{eqnarray*}
\end{enumerate}
\end{cor}
 
 \begin{proof}
It follows directly from theorem \ref{thm:LogK3Embed} and Lemma \ref{thm:K3RealLog}.  
 \end{proof}

\subsection*{The case of repeated eigenvalues}
After the previous result, it is natural to ask what can be said in the case of repeated eigenvalues. In this case, there are some theoretically interesting examples 
showing that: 
 \begin{enumerate}
  \item There are embeddable K3 matrices with repeated \emph{negative} eigenvalues. Remarkably, these matrices do not admit K3 Markov generators and their rates are not identifiable. See forthcoming Example \ref{ex35}.
 \item Restricted to the case of repeated positive eigenvalues, there are K3 Markov matrices for which rates are not identifiable. 
 \end{enumerate}
 
 Although the matrices presented in the forthcoming examples are close to saturation and have a big mutation rate, they have still biological interest by virtue of Theorem \ref{realism}. 
%
%

\paragraph{Repeated negative eigenvalues}
 It is well known that a matrix with negative eigenvalues has a real logarithm 
 if and only if the negative eigenvalues have even multiplicity \cite{Culver}. 
 Consider a K3 Markov matrix (actually, it is a $K2$ matrix) with eigenvalues $1$, $e^{-\lambda}$ and $-e^{-\mu}$ with multiplicity $2$, where $\lambda,\mu \geq 0$:
 \begin{eqnarray} \label{matrix_lm}
M=\frac{1}{4} \left ( \begin{array}{cccc}
 1+e^{-\lambda} - 2 e^{-\mu}  &  1+e^{-\lambda} + 2 e^{-\mu}  &  1-e^{-\lambda} & 1-e^{-\lambda}\\
 1+e^{-\lambda} + 2 e^{-\mu}  & 1+e^{-\lambda} - 2 e^{-\mu}  &  1-e^{-\lambda} & 1-e^{-\lambda}\\
  1-e^{-\lambda} & 1-e^{-\lambda} & 1+e^{-\lambda} - 2 e^{-\mu}  & 1+e^{-\lambda} + 2 e^{-\mu}\\
 1-e^{-\lambda} & 1-e^{-\lambda} & 1+e^{-\lambda} + 2 e^{-\mu}  & 1+e^{-\lambda} - 2 e^{-\mu}   \end{array} \right ) 
\end{eqnarray} 
By virtue of Theorem \ref{thm:LogK3Embed}, $M$ has no real logarithm with K3 or K2 form. 

According to Theorem 1.28 in \cite{higham}, distinct logarithms for the matrix $M$ are obtained  as $\overline{S} D \overline{S}^{-1}$, where $D$ is a diagonal matrix with distinct determinations of the logarithm of the eigenvalues of $M$ and the columns of $\overline{S}$ form a basis of eigenvectors of $M$ (see also Chap. VIII\S 8 in \cite{Gantmacher}). %
Hence, if we choose a pair of conjugated eigenvectors $v,\overline{v}$ of the eigenvalue $-e^{-\mu}$ as the third and fourth columns of the matrix $\overline{S}$, then all the matrices of the form
\begin{eqnarray*}
 Q_{k}= \overline{S} \; \dd\Big(0,-\lambda,-\mu +(2k+1)\pi i,-\mu -(2k+1)\pi i\Big) \; \overline{S}^{-1}, \qquad k\in \mathbb{Z}
\end{eqnarray*}
are real logarithms of $M$. 
For example, we can take 
\begin{small}
\[\overline{S}=\begin{pmatrix}
1 & 1 & 1+i & 1-i \\
1 & 1 & -1-i & -1+i  \\
1 & -1 & 1-i & 1+i  \\
1 & -1 & -1+i & -1-i \\
\end{pmatrix}\]
\end{small}
to obtain
  \begin{small}
 \[Q_k= \frac{1}{4} \begin{pmatrix}
   -\lambda-2\mu  & -\lambda+2\mu  & \lambda -2\pi (2k+1)& \lambda+2\pi (2k+1)\\
   -\lambda+2\mu   & -\lambda-2\mu  & \lambda + 2\pi (2k+1) & \lambda -2\pi (2k+1) \\
    \lambda +2\pi (2k+1) & \lambda -2\pi (2k+1) &  -\lambda-2\mu   & -\lambda+2\mu   \\
      \lambda-2\pi (2k+1)  & \lambda+2\pi (2k+1) & -\lambda+2\mu   & -\lambda-2\mu  \end{pmatrix}.\]
\end{small}
\noindent 
It is straightforward to check that this is a Markov generator if and only if $2\pi |2k+1| \leq \lambda$ and $\lambda \leq 2 \mu$. In particular, we see that $M$ is embeddable if 
\begin{eqnarray}\label{eq:neg_unid}
2 \pi \leq \lambda \leq 2\mu. 
\end{eqnarray}
Moreover, in this case, it is enough to take $k=0$ and $k=-1$ to obtain a pair of Markov generators for $M$. 

We illustrate this construction with a numerical example. 
\begin{ex}\label{ex35}
Let us take $\lambda =7$ and $\mu =4$, so that $2\pi \leq \lambda \leq 2\mu$. Then the matrix $M$ is (rounding off to 7 decimals) 
\begin{eqnarray}\label{K2neg}
\left ( \begin{array}{cccc}
0.2410701 & 0.2593858 & 0.2497720 & 0.2497720 \\
0.2593858 & 0.2410701 & 0.2497720 & 0.2497720 \\
0.2497720 & 0.2497720 & 0.2410701 & 0.2593858 \\
0.2497720 & 0.2497720 & 0.2593858 & 0.2410701
\end{array} \right ).
\end{eqnarray}
%
 \normalsize
As mentioned above, $Log(M)$ is not a real matrix and hence  is not a rate matrix either. In spite of that we are still able to find a pair of Markov generators for $M$ by taking $k=0$ and $k=-1$ respectively:
\begin{center}
$Q_{0}=\frac{1}{4} \begin{pmatrix}
   -15   & 1  & 7 -2\pi & 7 +2\pi \\
   1   & -15  & 7+2\pi  & 7-2\pi \\
   7+2\pi & 7-2\pi  & -15   & 1   \\
   7-2\pi & 7+2\pi & 1   & -15   \\
  \end{pmatrix} \qquad Q_{-1}= \frac{1}{4} \begin{pmatrix}
   -15   & 1  & 7 + 2\pi & 7 -2\pi \\
   1   & -15  & 7-2\pi  & 7+2\pi \\
   7-2\pi & 7+2\pi  & -15   & 1   \\
   7+2\pi & 7-2\pi & 1   & -15   \\
  \end{pmatrix} $
\end{center}
\end{ex}
 
 \vspace{3mm}
The previous example shows that a K3 Markov matrix $M$ can be embeddable, even if it does not have any Markov generator with K3 form (see Theorem \ref{thm:LogK3Embed} (i)). Thus, embeddability of $K3$ matrices does not imply $K3$-embeddability, and is not determined by the principal logarithm (cf. \cite{verbyla,Kaie}). 
This fact exhibits that the structure of the K3 model, which imposes certain symmetries between transitions and between transversions, is not always captured by the same symmetries between the mutation rates 
(cf. \cite{Kimura1980,kimura1981}). 
%
The reader may note that the  expected number of substitutions for a process ruled by a matrix $M$  as in (\ref{matrix_lm}) is $-\frac{1}{4}tr(Q_k)=\frac{1}{4}(\lambda+2\mu)\geq \pi.
$

\begin{rk} \label{multclosure_embed} \rm 
After the previous example, we derive easily that the product of embeddable matrices within the Kimura 3ST model is not necessarily embeddable (cf. Theorem \ref{closureK3embed}).
Indeed, it is enough to consider the embeddable matrix $M$ shown in (\ref{K2neg}) and any K3 matrix $N$ with positive eigenvalues $1,x,y,z$ satisfying the inequalities (\ref{ineq}) so that $N$ is embeddable. 
The product of $M$ and $N$ is clearly a K3 Markov matrix, whose eigenvalues are the product of the eigenvalues of $M$ and $N$. Thus, $M N$ has two \emph{different negative} eigenvalues and another positive eigenvalue. By virtue of \cite{Culver}, $MN$ has no real logarithm so, in particular, it cannot be embeddable. 
\end{rk}

\paragraph{Repeated positive eigenvalues}
A similar computation to that for negative eigenvalues can be done by assuming positive and repeated eigenvalues: $1$, $e^{-\lambda}$ and $e^{-\mu}$ with multiplicity 2. In this case, we can produce real logarithms by taking  
\begin{eqnarray*}
 Q_{k} & =  & \overline{S} \; \dd\Big(0,-\lambda,-\mu +2k \pi i,-\mu -2k \pi i\Big) \; \overline{S}^{-1} = \\
 & = &  \frac{1}{4} \begin{pmatrix}
   -\lambda-2\mu  & -\lambda+2\mu  & \lambda -4\pi k& \lambda +4\pi k\\
   -\lambda+2\mu   & -\lambda-2\mu  & \lambda +4\pi k & \lambda -4\pi k \\
    \lambda +4\pi k & \lambda -4\pi k &  -\lambda-2\mu   & -\lambda+2\mu   \\
     \lambda -4\pi k  & \lambda +4\pi k & -\lambda+2\mu   & -\lambda-2\mu  \end{pmatrix}.
\end{eqnarray*}
Such a matrix is a Markov generator if and only if $4\pi|k| \leq \lambda$ and $\lambda\leq 2\mu$. In particular, we see that if 
$0\leq \lambda  \leq 2\mu$, then $M$ is $K2$-embeddable since $Q_0$ is a Markov generator. If in addition, 
\begin{eqnarray}\label{eq:pos_unid}
4\pi \leq \lambda\leq 2\mu,
\end{eqnarray}
then we have (at least) three Markov generators, which correspond to $k=0,\pm 1$.

Note that the Markov generator $Q_0$ is a K3 matrix (corresponding to the principal logarithm) while $Q_1$ and $Q_{-1}$ are not.
%
%
The  expected number of substitutions for a process ruled by these matrices is 
 $-\frac{1}{4}tr(Q_k)=\frac{1}{4}(\lambda+2\mu)\geq 2\pi$. 

\begin{rk}\rm 
Following the construction of Theorem \ref{realism}, for any $n\geq 1$ 
the matrix $e^{(1/n)Q_0}$ is a K2 Markov matrix that approaches to $Id$ when $n$ grows. %
If a substitution process is ruled by such a matrix for some long time $t>0$, the rates of the resulting Markov matrix $e^{t (Q_0/n)}$ become unidentifiable at some point. 
This situation cannot occur for K3 embeddable matrices with different eigenvalues (see Theorem \ref{thm:LogK3Embed}).
\end{rk}

\vspace{3mm}
The existence of two or more Markov generators for the same Markov matrix (both in the case of negative and positive eigenvalues) exhibit that, in general, mutation rates are not identifiable from the mutation probabilities. 
Even more, if we restrict to the Kimura 3ST submodels, such matrices do not appear as marginal cases. 
Indeed, we have seen that identifiability of rates of a K2 embeddable matrix $M$ with eigenvalues $1,e^{-\lambda},e^{-\mu}$ (with multiplicity 2) does not hold in the subspace defined by the inequalities (\ref{eq:pos_unid}), which has positive measure within the space of all K2 Markov matrices. 
Furthermore, we have also seen that those Markov matrices with a negative repeated eigenvalue $1,e^{-\lambda},-e^{-\mu}$ (with multiplicity 2) satisfying (\ref{eq:neg_unid}) are embeddable but not K2-embeddable (nor K3-embeddable). The space of such matrices has also positive measure within space of K2 Markov matrices.

Similarly, for a Jukes-Cantor matrix with eigenvalues $1$ and $e^{-\lambda}$ (with multiplicity 3), identifiability of rates does not hold in the subspace defined by $4\pi \leq \lambda$, which has positive measure within the space of all JC matrices.
On the other hand, every embeddable matrix is $JC$-embeddable since by \cite{Culver}, a necessary condition for a Markov Jukes-Cantor matrix to be embeddable is that its eigenvalues $1,x$ are positive, and then  
Theorem \ref{thm:LogK3Embed} ensures that the principal logarithm (which is a JC matrix) is a Markov generator (it is enough to check that $x \geq x^2$, which is true since $x\in  [0, 1]$).

\vspace{3mm}

To conclude this section, we state the following theorem which summarizes the consequences of the previous examples:

\begin{thm}\label{thm:identif}
\begin{enumerate}
\item For the Kimura 3ST model, $K3$-embeddability and the identifiability of rates are generic properties of embeddable matrices.\footnote{A \emph{generic property} is a property that holds almost everywhere, that is, except for a set of measure zero.}

\item For the Kimura 2ST model, $K2$-embeddability and the identifiability of rates are not generic properties of  embeddable matrices. 

\item For the Jukes-Cantor model, every embeddable matrix is JC-embeddable but identifiability of rates is not a generic property of embeddable matrices. 

%
%
%
\end{enumerate}
\end{thm}

%

%
%
%

\section{The volume of embeddable K3 matrices}
Roughly speaking, the goal of this section is to measure how many K3 Markov matrices we are considering when the continuous-time approach is taken and compare this value with the corresponding value of K3 Markov matrices with no further restriction. These values are expressed in terms of the volume of the corresponding subspaces.  At the same time, it is direct to obtain the volume of subspaces of K3 Markov matrices with some constraints to make them biologically realistic.  To this aim, we proceed to represent K3 Markov matrices in a geometrical way as follows (cf. \cite{casfer_k3}): keeping the notation introduced in (\ref{coord}), 
Lemma \ref{Lema:K3Diag} allows us to identify the K3 Markov matrices with the coordinates $(x,y,z)$ of a 3-dimensional space. Moreover, since every K3 matrix is a convex combination of the identity matrix and permutation matrices:
\begin{eqnarray*}
\footnotesize
M=
a \; \begin{pmatrix}
  1 & 0 & 0 & 0\\
  0 & 1 & 0 & 0 \\
  0 & 0 & 1 & 0 \\
  0 & 0 & 0 & 1 \\
  \end{pmatrix}
 + b \; \begin{pmatrix}
  0 & 1 & 0 & 0\\
  1 & 0 & 0 & 0 \\
  0 & 0 & 0 & 1 \\
  0 & 0 & 1 & 0 \\
  \end{pmatrix}
+c \; \begin{pmatrix}
  0 & 0 & 1 & 0\\
  0 & 0 & 0 & 1 \\
  1 & 0 & 0 & 0 \\
  0 & 1 & 0 & 0 \\
  \end{pmatrix}
+ d\; \begin{pmatrix}
  0 & 0 & 0 & 1\\
  0 & 0 & 1 & 0 \\
  0 & 1 & 0 & 0 \\
  1 & 0 & 0 & 0 \\
  \end{pmatrix}, \quad 
  \begin{array}{c} a,b,c,d\geq 0,  \\ a+b+c+d=1
  \end{array}
\end{eqnarray*}
\normalsize
the space of all K3 Markov matrices describes the 3-dimensional simplex (a regular tetrahedron) with vertices given by the corresponding eigenvalues: $p_1=(1,1,1)$, $p_2=(1,-1,-1)$, $p_3=(-1,1,-1)$ and $p_4=(-1,-1,1)$ (see Figure \ref{simp}). The centroid of this simplex has coordinates (eigenvalues) $O=(0,0,0)$ and corresponds to the matrix 
$$M=\frac{1}{4} \begin{pmatrix}
  1 & 1 & 1 & 1\\
  1 & 1 & 1 & 1 \\
  1 & 1 & 1 & 1 \\
  1 & 1 & 1 & 1 \\
  \end{pmatrix}.$$
  According to this representation, the Jukes-Cantor matrices \cite{JC69} ($b=c=d$) correspond to the line determined by the identity vertex and the centroid of the simplex ($x=y=z)$, while Kimura 2ST matrices \cite{Kimura1980} ($c=d$) correspond to a plane section of the simplex ($y=z$). 

  \begin{figure}[h]
\begin{center}
\includegraphics{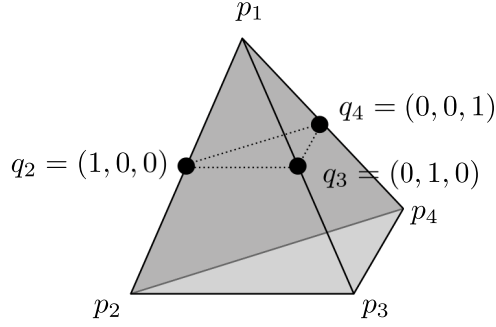}
\caption{\label{simp} Simplex representing all K3 Markov matrices. Each matrix is represented by its eigenvalues.}
\end{center}
\end{figure}

We proceed to compute the volume of a number of subspaces of K3 Markov matrices that can be interesting from a biological perspective. 
First of all, we introduce some notation:
\begin{eqnarray*}
& \Delta  & \mbox{is the space of all K3 Markov matrices}. \\
& \Delta_{*} & \mbox{is the space of matrices for which $a\geq b,c,d$}. \\
& \Delta_{+} & \mbox{is the subspace of matrices with only positive eigenvalues}. \\
& \Delta_d & \mbox{is the subspace of matrices that are diagonally dominant, i.e. }a\geq b+c+d. \\
& \Delta_{\varepsilon} & \mbox{is the subspace of embeddable K3 matrices}. 
\end{eqnarray*}
By Lemma \ref{Lema:K3Diag}, it is straighforward to check that $\Delta_d \subset \Delta_+ \subset \Delta_* \subset \Delta$. Note that the examples of the preceding section show that $\Delta_{\varepsilon} \not\subset \Delta_*$ (see the matrix in (\ref{K2neg})). However, according to Corollary \ref{cor:K3noRepeated}, all embeddable matrices that are not in $\Delta_+$ correspond to the case of repeated eigenvalues, so they are a marginal case with measure (i.e. volume) 0 within the whole space of K3 matrices. This is because these matrices are constrained by nontrivial algebraic constraints, which make the dimension of the corresponding subspace necessarily smaller. 
Nevertheless, as shown in (e) and (f) of the next result,  there are lots of embeddable matrices with no repeated eigenvalues that are not diagonal dominant. 
  
\begin{thm}\label{thm:vols}
We have the following: \\
(a) $V(\Delta_{d})=1/3$;  
(b) $V(\Delta_{+})=1/2$; 
(c) $V(\Delta_{*})=2/3$; 
(d) $V(\Delta)=8/3$; 
(e) $V(\Delta_{\EE})=1/4$; \\
(f) $V(\Delta_{\varepsilon} \cap \Delta_d)\simeq 0.20336$ 
\end{thm}

\begin{proof}
\noindent (a) The space $\Delta_d$ of diagonal dominant matrices is defined by the inequality $a\geq b+c+d$. Since $a+b+c+d=1$, this is equivalent to $a\geq 1/2$. Thus, $\Delta_d$ is the regular simplex with vertices $p_1$, $q_2$, $q_3$ and $q_4$ (see Figure \ref{simp}), and its volume is given by the well known formula
\begin{eqnarray*}
V(\Delta_d)=\frac{1}{6} 
\, | det(\overrightarrow{p_1q_2}, \overrightarrow{p_1q_3}, \overrightarrow{p_1q_4}) | = 1/3.
\end{eqnarray*}

\vspace{2mm}
\noindent (b) The space $\Delta_+$ of matrices with positive eigenvalues is composed of  $\Delta_d$ together with the simplex with vertices $q_1$, $q_2$, $q_3$ and the centroid $O$. The formula above gives that this last simplex has volume $1/6$. Therefore, $V(\Delta_+)=V(\Delta_d)+1/6=1/2$. 

%
\vspace{2mm}
\noindent (c) The three inequalities defining $\Delta_*$ are equivalent to $x+y\geq 0$, $x+z\geq 0$ and $y+z\geq 0$. If we denote by $p_{ijk}$ the centroid of the triangle defined by $p_i$, $p_j$ and $p_k$, it is straighforward to see that $\Delta_*$ is composed of $\Delta_+$ together with the three simplices defined by $\{O,q_2,q_3,p_{123}\}$, 
$\{O,q_3,q_4,p_{134}\}$ and $\{O,q_4,q_2,p_{124}\}$. These three simplices have the same volume: 
\begin{eqnarray*}
\frac{1}{6} 
\, | det(
\overrightarrow{Oq_i}, 
\overrightarrow{Oq_j},
 \overrightarrow{Op_{1ij}}
) |
 = 1/18.
\end{eqnarray*}
It follows that $V(\Delta_*)=1/2+3\, (1/18)=2/3$. 

\vspace{2mm}
\noindent (d) follows similarly to the computation of (a).

\vspace{2mm}
\noindent (e) The computation of $\Delta_{\varepsilon}$ is more involved. First of all, as noted above, we can restrict the computation to matrices with no repeated eigenvalues. 
Therefore, it is enough to compute the volume of the space $\EE$ defined by the inequalities (\ref{ineq}). By cutting the space $\EE$ with planes of the form $z=a$ with $a\in [0,1]$, we obtain a shape like Figure \ref{int}. The computation follows by integrating the corresponding area for all values of $a$.
Given $a\in [0,1]$ and $x\in [0,a]$, the range of values for $y$ is between $a\,x$ and $x/a$, while for $x\in [a,1]$, the value of $y$ lies between $a\,x$ and $a/x$. 
We are led to compute the following integral
\begin{eqnarray*}
 V(\Delta_{\EE})=\int_0^1 \left ( \int_0^z \int_{zx}^{x/z} dy\; dx + \int_z^1 \int_{xz}^{z/x} dy \; dx \right ) dz
\end{eqnarray*}
which can be easily shown to be equal to $1/4$.

\vspace{2mm}
\noindent (f) The computation of the volume of the intersection $\Delta_d \cap \Delta_{\varepsilon}$ is similar to the computation of $V(\Delta_{\varepsilon})$ but for each plane section $z=a$, $a\in [0,1]$ we have to remove the area of the space below the line $x+y=1-a$ (corresponding to non-embeddable matrices). This leads to two different situations: for $a\in [0,3-2\sqrt{2}]$ the line cuts the hyperbola $xy=a$; while for $a\in [3-2\sqrt{2},1]$ line and hyperbola does not meet. The computation of the corresponding integrals is tedious and we do not include it here. The final value has been obtained using the mathematical software SAGE \cite{sagemath}.
\end{proof}

\begin{figure}
\begin{center}
\includegraphics{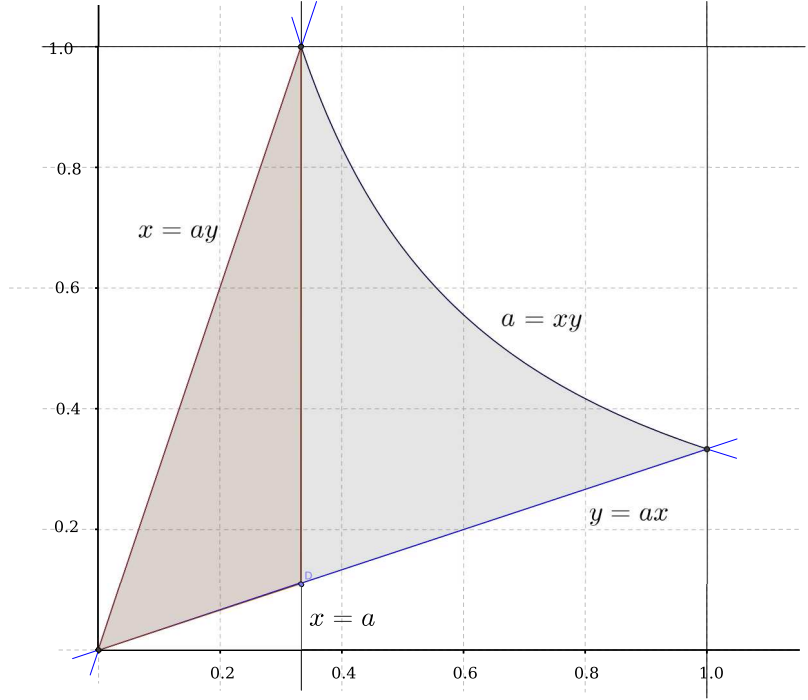}
\caption{\label{int} Plane section of the space $\varepsilon$ of embeddable K3 matrices with the plane $z=a, a\in [0,1]$.}
\end{center}
\end{figure}

The values of these volumes illustrate the relative size between the spaces of Markov matrices considered above. 
Table \ref{tab:volume} shows these volumes and the relative volume of embeddable matrices in each of the above subspaces of K3 Markov matrices.  These relative volumes are a measure of how many matrices are rejected when taking the continuous-time approach instead of considering other subspaces of matrices. 
%
%
The figures in the second row of the table show, in particular, that there is a big difference between these volumes.
For example, embeddable matrices suppose half of the matrices with positive eigenvalues. Similarly,  only three out of eight K3 Markov matrices satisfying $a\geq b,c,d$ are embeddable, while we observe that  they represent more than the 60\% of the diagonal dominant matrices. However, if we only consider diagonal dominant matrices we are rejecting a non-negligible number of embeddable matrices (the difference of volumes between embeddable and diagonal dominant embeddable is $1/4-0.20336=0.046641$, see Theorem \ref{thm:vols}).

\begin{table}[h!]
  \centering
    \begin{tabular}{c||c|cccc}
  & $\Delta_{\varepsilon}$ & $\Delta$ & $\Delta_*$ & $\Delta_+$ & $\Delta_d$  \\ \hline
$V(\cdot)$ & 1/4 & 8/3 & 2/3 & 1/2 & 1/3  \\
relative vol. of embeddable   & 1 & 3/32 & 3/8 & 1/2 &  0.61008
 \\
  \end{tabular}
    \caption{\label{tab:volume} Volumes and relative volumes of the embeddable K3 matrices. The relative volumes of embeddable matrices within each spaces are shown in the second row of the table and are obtained as the quotients $V(\Delta_{\varepsilon} \cap \cdot)/V(\Delta_{\cdot})$. }
\end{table}

\begin{rk} \rm
From Theorem \ref{thm:identif} (i) it follows that 
all the values of Table \ref{tab:volume} remain the same if we only consider $K3$-embeddable matrices (instead of embeddable matrices). 
%
\end{rk}

Since we have no characterization for embeddable matrices with repeated eigenvalues, and this set has positive measure within the Kimura 2ST model (Theorem \ref{thm:identif}), we are not able to compute volumes within this model. 
%
%

Nevertheless, for the Jukes-Cantor model, Theorem \ref{thm:identif} states that embeddability is equivalent to JC-embeddability, which has been characterized in Theorem \ref{thm:LogK3Embed}. 
Adapting the notation used for the K3 matrices to the JC model, we get that 
\begin{eqnarray*}
 \Delta^{JC} & = &\{(x,x,x)\in \Delta \mid x\in [-1/3,1]\}; \\
 \Delta^{JC}_{*} & = & \{(x,x,x)\in \Delta \mid x\in [0,1]\}; \\
 \Delta^{JC}_{+}=\Delta^{JC}_{\varepsilon}& = &\{(x,x,x)\in \Delta \mid x\in (0,1]\};\\
 \Delta^{JC}_{d}& = &\{(x,x,x)\in \Delta \mid x\in [1/3,1]\}.
\end{eqnarray*}
A straightforward computation shows that the space of embeddable Jukes-Cantor matrices has volume (length) $\sqrt{3}$ while the space of all Markov Jukes-Cantor matrices has volume $\frac{4}{3}\sqrt{3}$. 
That is, three out of four Jukes-Cantor matrices are embeddable.

\section{Discussion}

As suggested in the introduction, there are four connected problems relative to the algebraic and continuous-time evolutionary models. Namely, 
\begin{enumerate}
 \item whether a given Markov matrix is $e^Q$ for some rate matrix $Q$ (embedding problem);
 \item whether a given Markov matrix $M$ within an algebraic model $\mathcal{M}$ is $e^Q$ for some rate matrix $Q$ with the same symmetries as $M$ ($\mathcal{M}$-embedding problem);
 \item whether the product of embeddable matrices is embeddable, i.e. if $Q_1$ and $Q_2$ are rate matrices, then is it true that $e^{Q_1} e^{Q_2} =e^{Q_3}$ for some rate matrix $Q_3$? (multiplicative closure of embeddable matrices);
 \item same question as in 3. but with $Q_1$,$Q_2$ and $Q_3$ within a particular continuous-time model (multiplicative closure of $\mathcal{M}$-embeddable matrices).
\end{enumerate}
The first two questions are motivated by the connection between the algebraic and the continuous-time models. The last two are more intrinsic of the continuous-time approach.
In this paper, we have discussed the embedding problem  for Markov matrices within the Kimura 3-parameter model (\textbf{1.}). 
Under a generic assumption (that is, eigenvalues should be different), we have obtained a characterization of the embeddability of the matrices for this model in terms of some inequalities relative to their eigenvalues (Corollary \ref{cor:K3noRepeated}). Moreover, Theorem \ref{thm:identif} shows that in this case, rates are identifiable and keep the K3 symmetries, so that 
embeddability holds if and only if so does $K3$-embeddability, which has been characterized in Theorem \ref{thm:LogK3Embed} (\textbf{2.}).  
%
%
As for the problem (\textbf{3.}), Remark \ref{multclosure_embed} shows that the  product of embeddable matrices within the Kimura 3ST is not embeddable in general. However, under the assumption of different eigenvalues again, problems (\textbf{3.}) and (\textbf{4.}) become equivalent and Theorem \ref{closureK3embed} 
gives an affirmative answer to both questions. 

As a consequence of the characterization of Corollary \ref{cor:K3noRepeated}, we have been able to compute and compare the volume of embeddable and K3-embeddable matrices within some subspaces of K3 Markov matrices that may be regarded as a compromise solution between the continuous-time approach and the whole space of K3 Markov matrices (see Theorem \ref{thm:vols} and Table \ref{tab:volume}).

Despite the fact that the evolutionary submodels of Kimura 3ST (Kimura 2ST and Jukes-Cantor) have repeated eigenvalues, the results obtained here also give information about them.
While a characterization of the the embeddability of the Jukes-Cantor model follows easily from the general case of K3 matrices, the embeddability of Kimura 2ST model is more involved and remains an open question. 
Under this model, we have provided examples of matrices with repeated eigenvalues showing that there are embeddable matrices within the K2ST model that have no Markov generator with K3 form (see Theorem \ref{thm:identif}) and although these matrices are close to saturation, Theorem \ref{realism} shows they are still biologically relevant. 

 Another open last issue is the identifiability of rates for K3 embeddable matrices with repeated eigenvalues. In this case, the inequalities (\ref{eq:neg_unid}) for negative eigenvalues and (\ref{eq:pos_unid}) for positive eigenvalues are sufficient for a K3 Markov matrix to have different Markov generators. 
We believe that the these inequalities are actually necessary, but a much more technical analysis is required.  
Moreover, we conjecture that under the Kimura 3ST model, the rates of an embeddable matrix that is not $K3$-embeddable are not identifiable. 
An affirmative answer to these questions would allows us to 
characterize the embeddability of any K3 Markov matrix (cf. Theorem \ref{thm:K3RealLog}), and show that for Markov matrices with positive eigenvalues, embeddability is equivalent to $\mathcal{M}$-embeddability for K3ST and any of its submodels.
We defer such a study for a future publication.

\section{Acknowledgement}
The authors want to thank Marta Casanellas and Jeremy Sumner for conversations held on this topic. They alse wish to thank the reviewers for useful comments and suggestions. 

JRL and JFS are partially supported by Spanish government MTM2015-69135-P.  Second author is also supported by Generalitat de Catalunya 2014SGR634.

\setlength{\parindent}{0pt}
\bibliographystyle{alpha}
\newcommand{\etalchar}[1]{$^{#1}$}

\end{document}